\documentclass[a4paper,10pt]{article}

\title{A Concrete Representation of Observational Equivalence for PCF}
\author{Martin Churchill, James Laird and Guy McCusker}
\date{March 2009}

\usepackage{stmaryrd} \usepackage{diagrams} \usepackage{bussproofs}

\usepackage{amsmath,amssymb, amsfonts} \usepackage{hyperref}
\usepackage{diagrams} \usepackage{bussproofs}
\usepackage{graphicx} \usepackage{pstricks,pst-node}
\usepackage{stmaryrd} \usepackage{verbatim} \usepackage{tikz}

\newtheorem{theorem}{Theorem}[section]

\newtheorem{proposition}[theorem]{Proposition}

\newenvironment{proof}[1][Proof]{\begin{trivlist}
\item[\hskip \labelsep {\bfseries #1}]}{\end{trivlist}}
\newenvironment{definition}[1][Definition]{\begin{trivlist}
\item[\hskip \labelsep {\bfseries #1}]}{\end{trivlist}}
\newenvironment{example}[1][Example]{\begin{trivlist}
\item[\hskip \labelsep {\bfseries #1}]}{\end{trivlist}}

\newcommand{\qed}{\nobreak \ifvmode \relax \else
  \ifdim\lastskip<1.5em \hskip-\lastskip \hskip1.5em plus0em
  minus0.5em \fi \nobreak \vrule height0.75em width0.5em
  depth0.25em\fi}

\begin{document}
\maketitle

\begin{abstract}
  A result on observational equivalence for PCF and innocent
  strategies, as presented at the \emph{Games for Logic and
    Programming Languages} (GaLoP) workshop in York, March 2009.

  The full abstraction result for PCF using game semantics requires
  one to identify all innocent strategies that are \emph{innocently
    indistinguishable}. This involves a quantification over all
  innocent tests, cf. quantification over all innocent contexts. Here
  we present a representation of innocent strategies that equates
  innocently indistinguishable ones, yielding a representation of PCF
  terms that equates precisely those terms that are observational
  equivalent.
\end{abstract}

\section{Introduction}

In recent years game semantics has provided an accurate model for
various programming languages, leading to the first full abstraction
results for a variety of languages and in a unified way
\cite{AMc_GS}. In such models programs are interpreted as
\emph{strategies}, i.e. highly constrained (sets of) processes; and
adding semantic power corresponds to weakening restrictions on such
sets.

An early success was providing a the first fully abstract model of PCF
\cite{AJM_PCF, HO_PCF}, answering a challenge first posed in
\cite{Plot_LCF}. In this case we need to restrict our strategies to
representing pure functions, using a construct known as
\emph{innocence} which states that the strategy isn't allowed to
depend on the entire history (state) but only part of it; a relevant
context. So on the one hand we are dealing with intensional processes;
but on the other hand we are dealing with pure functions (albeit only
the sequential ones). This conflict rears its head in the full
abstraction result for PCF, where observational equivalence only holds
when one identifies strategies that cannot be distinguished by an
innocent test.

\begin{example}
We shall define two candidate innocent strategies for addition,
$\mathsf{add_{LR}}$ and $\mathsf{add_{RL}}$ over the game $\mathbf{N}
\times \mathbf{N} \Rightarrow \mathbf{N}$, which evaluate their
arguments left-to-right and right-to-left respectively. Let $q$
represent the unique \textsf{O}-question in the game $\mathbf{N}$, and
$m$ and $n$ range over the natural numbers. Maximal plays of
$\mathsf{add_{LR}}$ are then of the form:

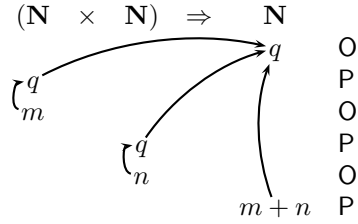
\[
  \begin{array}{cccccl}
   (\mathbf{N} & \times & \mathbf{N}) & \Rightarrow & \mathbf{N}  \\
  & & & & \Rnode{a}{q} & \mathsf{O} \\
  \Rnode{d}{q} & & & & & \mathsf{P} \\
  \Rnode{e}{m} & & & & & \mathsf{O}\\
  & & \Rnode{b}{q} & & & \mathsf{P} \\
  & & \Rnode{c}{n} & & & \mathsf{O} \\
  & & & & \Rnode{f}{m+n} & \mathsf{P} \\
  \end{array}
\ncarc[nodesep=1pt,arcangle=-20]{<-}{a}{f}
\ncarc[nodesep=1pt,arcangle=-20]{<-}{a}{b}
\ncarc[nodesep=1pt,arcangle=-80]{<-}{b}{c}
\ncarc[nodesep=1pt,arcangle=-80]{<-}{d}{e}
\ncarc[nodesep=1pt,arcangle=-20]{<-}{a}{d}
\]

\noindent Maximal plays of $\mathsf{add_{RL}}$ are of the form:

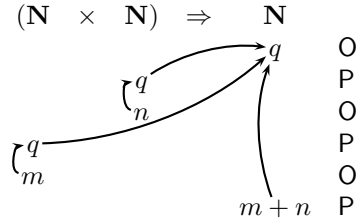
\[
  \begin{array}{cccccl}
   (\mathbf{N} & \times & \mathbf{N}) & \Rightarrow & \mathbf{N}  \\
  & & & & \Rnode{a}{q} & \mathsf{O} \\
  & & \Rnode{b}{q} & & & \mathsf{P} \\
  & & \Rnode{c}{n} & & & \mathsf{O} \\
  \Rnode{d}{q} & & & & & \mathsf{P} \\
  \Rnode{e}{m} & & & & & \mathsf{O}\\
  & & & & \Rnode{f}{m+n} & \mathsf{P} \\
  \end{array}
\ncarc[nodesep=1pt,arcangle=-20]{<-}{a}{f}
\ncarc[nodesep=1pt,arcangle=-20]{<-}{a}{b}
\ncarc[nodesep=1pt,arcangle=-80]{<-}{b}{c}
\ncarc[nodesep=1pt,arcangle=-80]{<-}{d}{e}
\ncarc[nodesep=1pt,arcangle=20]{<-}{a}{d}
\]
\end{example}

We note that the strategies $\mathsf{add_{RL}}$ and
$\mathsf{add_{LR}}$ are not equal. However they are not
distinguishable by any innocent test --- for any innocent strategy
$\alpha : (\mathbb{N} \times \mathbb{N} \Rightarrow \mathbb{N})
\Rightarrow \Sigma$ we have $\mathsf{add_{LR}};\alpha =
\mathsf{add_{RL}};\alpha$. If we identify such
innocently-indistinguishable strategies, we factor out information
such as the number of times and order arguments are interrogated,
which are irrelevant details in a functional world; and it is with
respect to this identification that the full abstraction results for
PCF in \cite{AJM_PCF, HO_PCF} hold. However, quantifying over all
innocent strategies doesn't seem far from quantifying over all
innocent contexts, so it would be better if we could get a more
concrete handle on this observational preorder. Loader's result
\cite{Loa_PCF} places some restrictions on this: it was shown that
observational equivalence of PCF and finite base types is
undecidable. But nonetheless a more concrete presentation can be
given: here we introduce a candidate concrete representation of
innocent strategies (or PCF terms), and we define a map from innocent
strategies into this structure. This map identifies precisely those
strategies that are observationally equivalent. We believe this can be
used to construct a fully abstract model of PCF explicitly, with no
need of such a quotienting.

\section{Main Result}

\subsection{Views and Duality}

We recall standard definitions from game semantics of arena, justified
sequence, play, function space, strategy etc from
e.g. \cite{AMc_GS}. In particular we recall the definition of
\emph{\textsf{O}-view} and \emph{\textsf{P}-view}:

\begin{definition}
We define the \emph{\textsf{P}-view} of a play by
\begin{itemize}
 \item $\ulcorner \epsilon \urcorner = \epsilon$
 \item $\ulcorner s p \urcorner = \ulcorner s \urcorner p$ where $p$
   is a \textsf{P}-move
 \item $\ulcorner s i \urcorner = i$ where $i$ is an initial move
 \item $\ulcorner s p t o \urcorner = \ulcorner s \urcorner p o$,
   where \textsf{P}-move $p$ is the justifier of \textsf{O}-move $o$
\end{itemize}
We define the \emph{\textsf{O}-view} of a play
\begin{itemize}
 \item $\llcorner \epsilon \lrcorner = \epsilon$
 \item $\llcorner s o \lrcorner = \llcorner s \lrcorner o$ where $o$
   is an \textsf{O}-move
 \item $\llcorner s o t p \lrcorner = \llcorner s \lrcorner o p$,
   where \textsf{O}-move $o$ justifies \textsf{P}-move $p$
\end{itemize}
\end{definition}

\noindent We also recall the definition of the game $\Sigma =
(\{q,a\},\{q \mapsto OQ, a \mapsto PA\},\{\ast \vdash q \vdash a \},
\{\epsilon, q, qa \})$ and note that there are two strategies on this
game, $\top = \{ \epsilon, qa \}$ and $\bot = \{ \epsilon \}$. The
game $\Sigma$ allows us to note a duality between \textsf{O}-views and
\textsf{P}-views, since a single-threaded play in $A \rightarrow
\Sigma$ consists of a play in $A$ with the roles of \textsf{P} and
\textsf{O} reversed. This is useful to us because of the following
lemma:

\begin{proposition}
If $s$ is an \textsf{O}-view in the game $A$ then $qs$ is a
\textsf{P}-view in $A \rightarrow \Sigma$. If $qs$ is \textsf{P}-view
in $A \rightarrow \Sigma$ and $s$ is a play in $A$, then $s$ is an
\textsf{O}-view in $A$.
\label{Opviews}
\end{proposition}
\begin{proof}
\textsf{O}-views in $A$ are precisely the plays of the form
$o_1p_1o_2p_2o_3p_3 \ldots o_n(p_n)$ where the justifier of $p_i$ is
$o_i$. A \textsf{P}-view in $A \rightarrow \Sigma$ ending in $A$ must
be of the form $qp_1o_1p_2o_2 \ldots p_n(o_n)$ with each $o_i$
justified by the preceding $p_i$, and both are a move in $A$. Then
$p_1o_1 \ldots p_n(o_n)$ is a play in $A$ and since the parities are
reversed this is of the form $o_1p_1 \ldots o_n(p_n)$ with each $p_i$
justified by the preceding $o_i$, hence a \textsf{O}-view in
$A$. Clearly also any such \textsf{O}-view in $A$ yields a
\textsf{P}-view in $A \rightarrow \Sigma$ in this manner. \qed
\end{proof}

\begin{proposition}
Let $q_\Sigma s$ be a play in $A \rightarrow \Sigma$ ending in
$A$. Then $\ulcorner q_\Sigma s \urcorner = q_\Sigma \llcorner s
\lrcorner$, in the latter case taking the view with respect to the
arena $A$.
\label{ovS}
\end{proposition}
\begin{proof}
Induction on $s$. Base case $s = q_A$. Then $\ulcorner q_\Sigma q_A
\urcorner = q_\Sigma q_A = q_\Sigma \llcorner q_A \lrcorner$.

Inductive step --- if $s = s' p$ then $\ulcorner q_\Sigma s' p
\urcorner = \ulcorner q_\Sigma s' \urcorner p = q_\Sigma \llcorner s'
\lrcorner p = q_\Sigma \llcorner s' p \lrcorner$ since $p$ is an
\textsf{O}-move in the game $A$. If $s = s' p s'' o$ with $p$
justifying $o$ then $\ulcorner q_\Sigma s' p s'' o \urcorner =
\ulcorner q_\Sigma s' \urcorner p o = q_\Sigma \llcorner s' \lrcorner
p o = q_\Sigma \llcorner s' p s'' o \lrcorner$ since $(p,o)$ are
(\textsf{O},\textsf{P}) moves respectively in the game $A$. \qed
\end{proof}

\noindent In particular we will use this to note that innocent
strategies for $A \rightarrow \Sigma$ (i.e. innocent tests for $A$)
coincide with \textsf{O}-view functions on $A$. Further it is known
\cite{McC_FPC} that Linear Tests Suffice, so we only need consider such
\textsf{O}-view functions that deal with a single thread.

\begin{definition}
A set $S$ of well-bracketed \textsf{O}-views over an arena $A$ is
\emph{\textsf{O}-deterministic} if $so_1, so_2 \in S$ implies $o_1 =
o_2$, each $s \in S$ is single-threaded, each $s \in S$ begins with
the same initial move.
\end{definition}

\begin{definition}
If $S$ is an \textsf{O}-deterministic set over the arena $A$, we can
define the innocent strategy $\alpha_S : A \rightarrow \Sigma$ as a
\textsf{P}-view function $\alpha_S = \{ (q_\Sigma s, p) : sp \in S \}
\cup \{ (q_\Sigma t, a_\Sigma) : t \in S \wedge \mathsf{complete}(t)
\}$
\end{definition}

\noindent The above uses the observation that if $qs$ is a
\textsf{P}-view in $A \rightarrow \Sigma$ that does not end in
$a_\Sigma$ then $s$ is an \textsf{O}-view in $A$. We note that such
strategies yield well-bracketed plays since the \textsf{O}-views in
$S$ are well-bracketed, hence are the combination of
\textsf{P}-view/move pairs found in $\alpha_S$. We will soon show that
innocent tests on $A$ of the form $\alpha_S$ are the only ones needed
to distinguish two observationally inequivalent strategies; where
observational inequivalence comes from the following preorder:

\begin{definition}
Let $\sigma, \tau : A$ be innocent strategies. We write $\sigma
\leq_{ib} \tau$ if for any innocent $\alpha : A \rightarrow \Sigma$ if
$\sigma;\alpha = \top$ then $\tau;\alpha = \top$.
\end{definition}

\subsection{\textsf{O}-view Sets}

We shall now use some of these ideas to show that two innocent
strategies are observationally equivalent if and only if their sets of
\textsf{O}-views of prefixes of complete single-threaded plays are the
same.

\begin{definition}
A play $s$ is \emph{\textsf{O}-innocent} if for $s_1o_1, s_2o_2
\sqsubseteq s$ with $\llcorner s_1 \lrcorner = \llcorner s_2
\lrcorner$ and $o_1, o_2$ \textsf{O}-moves, we have $o_1 = o_2$. A
play $s$ is \emph{\textsf{P}-innocent} if for $s_1p_1, s_2p_2
\subseteq s$ with $\ulcorner s_1 \urcorner = \ulcorner s_2 \urcorner$
and $p_1, p_2$ \textsf{P}-moves, we have $p_1 = p_2$.
\end{definition}

\noindent We note that in a world of innocent strategies alone, a
strategy is equivalent to its set of \textsf{O}-innocent traces (since
after all, these are the only traces that can be ``realised'' by an
innocent opponent). It is also clear that all plays in an innocent
strategy are \textsf{P}-innocent.

\begin{definition}
Given a play $s$, define $\mathsf{ovw}(s) = \{ \llcorner t \lrcorner :
t \sqsubseteq s \}$.
\end{definition}

\begin{proposition}
If $s$ is a single-threaded \textsf{O}-innocent complete play,
$\mathsf{ovw}(s)$ is \textsf{O}-deterministic.
\label{Odet}
\end{proposition}
\begin{proof}
Suppose $so_1, so_2 \in \mathsf{ovw}(s)$. Then $so_1, so_2 = \llcorner
s'_1 \lrcorner, \llcorner s'_2 \lrcorner$. In practice we know that
$\llcorner s'_1 \lrcorner = \llcorner s_1 \lrcorner o_1$ and
$\llcorner s'_2 \lrcorner = \llcorner s_2 \lrcorner o_2$ with
$\llcorner s_1 \lrcorner = s = \llcorner s_2 \lrcorner$. But then
$s_1, s_2 \sqsubseteq s$ so $o_1 = o_2$ by \textsf{O}-innocence of
$s$.

We know that each $s' \in \mathsf{ovw}(s)$ is single-threaded, since
the \textsf{O}-view of a prefix of a single-threaded play is also
single-threaded.

We know that each $s' \in \mathsf{ovw}(s)$ begins with the same
initial move, since each $s'$ is the \textsf{O}-view of a prefix of
$s$ and as such must begin with the initial move of $s$ (since the
\textsf{O}-view of a play contains its first move). \qed
\end{proof}

\noindent We can now substantiate our remark above regarding
observational equivalence:

\begin{proposition}
$\sigma \leq_{ib} \tau$ iff for any \textsf{O}-deterministic set $S$
  on $A$ we have $\sigma;\alpha_S = \top$ implies $\tau;\alpha_S =
  \top$
\label{obseqv}
\end{proposition}
\begin{proof}
Clearly if $\sigma \leq_{ib} \tau$ the RHS holds by innocence of
$\alpha_S$.

Conversely, if $\sigma \leq_{ib} \tau$ does not hold then we have
$\alpha$ such that $\sigma;\alpha = \top$ and $\tau;\alpha = \bot$. By
Linear Tests Suffice we may assume that $\alpha$ consists only of
plays that interrogate their argument once, i.e. plays which are
single-threaded when restricted to $A$. Thus we have an interaction
sequence $s$ with $s = q_\Sigma s' a_\Sigma \in \alpha$ with $s' \in
\sigma$. Since $s \in \alpha_S$ and $\alpha_S$ is innocent we know $s$
must be \textsf{P}-innocent. By \ref{Opviews} it follows that $s'$
must be \textsf{O}-innocent. Further we know $s'$ is single-threaded
and complete (by well-bracketedness), and so $\mathsf{ovw}(s')$ is
\textsf{O}-deterministic by \ref{Odet}. Also, it is clear that $s \in
\alpha_{\mathsf{ovw}(s')}$. Thus we have
$\sigma;\alpha_{\mathsf{ovw}(s')} = \top$. Since
$\alpha_{\mathsf{ovw}(s')} \subseteq \alpha$ and $\tau;\alpha = \bot$
it follows that $\tau;\alpha_{\mathsf{ovw}(s')} = \bot$ since
composition is monotonic. Hence RHS does not hold in the case that $S
= \mathsf{ovw}(s')$. \qed
\end{proof}

\noindent We now formally define the set of observations over a
strategy $\sigma$, as the \textsf{O}-views of the prefixes of the
complete, single-threaded, \textsf{O}-innocent plays.

\begin{definition}
Given an innocent strategy $\sigma$, define $\mathsf{obs}(\sigma) = \{
\mathsf{ovw}(s) : s \in \sigma \wedge \mathsf{complete}(s) \wedge
\mathsf{Oinnocent}(s) \wedge \mathsf{singlethreaded}(s) \}$
\end{definition}

\noindent We thus have two constructions, $\mathsf{obs}$ that takes an
innocent strategy and returns a set of \textsf{O}-view sets, and $S
\mapsto \alpha_S$ which takes an \textsf{O}-deterministic set and
returns an innocent strategy. We can relate these constructions.

\begin{proposition}
  Let $S$ be an \textsf{O}-deterministic set on $A$ and $\sigma$ and
  innocent strategy on $A$. Then $\sigma;\alpha_S = \top$ if and only
  if $S \supseteq T \in \mathsf{obs}(\sigma)$
\label{topp}
\end{proposition}
\begin{proof}
Suppose $\sigma;\alpha_S = \top$. Then exists interaction sequence $q
s a$ with $s \in \sigma$ complete; and such that if $to \sqsubseteq s$
then $\llcorner t \lrcorner o \in S$. Thus $\mathsf{ovw}(s) \subseteq
S$. But $\mathsf{ovw}(s) \in \mathsf{obs}(\sigma)$ since $s \in
\sigma$ is complete (well-bracketedness), \textsf{O}-innocent (since
$qsa \in \alpha_S$ is \textsf{P}-innocent), and single-threaded (since
$S$ is \textsf{O}-deterministic) so $S \supseteq \mathsf{ovw}(s) \in
\mathsf{obs}(\sigma)$ as required.

Conversely, if $S \supseteq T \in \mathsf{obs}(\sigma)$ then $T =
\mathsf{ovw}(t)$ for some complete, \textsf{O}-innocent,
single-threaded play $t \in \sigma$. Consider the play $q_\Sigma t
a_\Sigma$ in $A \rightarrow \Sigma$. To show that $\sigma;\alpha_S =
\top$ it will suffice to show that $q_\Sigma t a_\Sigma \in
\alpha_S$. To see this we need to check that for all $t'$ with
$q_\Sigma t' p \sqsubseteq t$, $(\ulcorner q_\Sigma t' \urcorner, p)
\in \alpha_S$ where $t'$ is an even length sequence. If $p = a_\Sigma$
then we must have $t' = t$ and \ref{ovS} tells us that $(\ulcorner
q_\Sigma t \urcorner, a_\Sigma) = (q_\Sigma \llcorner t \lrcorner, a)
\in \alpha_S$ since $\llcorner t \lrcorner$ is both complete and in
$\mathsf{ovw}(t) = T$, and hence $S$. If $p$ is a move in $A$ then
$(\ulcorner q_\Sigma t' \urcorner, p) = (q_\Sigma \llcorner t'
\lrcorner, p) \in \alpha_S$ since $\llcorner t' \lrcorner p =
\llcorner t' p \lrcorner$ is in $\mathsf{ovw}(t) = T$ (and hence $S$).
Hence $q_\Sigma t a_\Sigma \in \alpha_S$ after all, giving us the
interaction sequence witness yielding $\sigma;\alpha_S = \top$. \qed
\end{proof}

\subsection{Full Abstraction}

In order to show that $\sigma =_{ib} \tau$ iff $\mathsf{ovw}(\sigma) =
\mathsf{ovw}(\tau)$, we first show an inequational version. The
observational preorder does not correspond to the subset ordering;
instead it corresponds to the following ordering:

\begin{definition}
Suppose $\sigma$ and $\tau$ are sets of \textsf{O}-deterministic sets
of \textsf{O}-views over an arena $A$. Write $\sigma \leq_{os} \tau$
if $\forall S \in \sigma \exists T \in \tau$ with $T \subseteq S$.
\end{definition}

\noindent It is clear that $\leq_{os}$ is a preorder.

\begin{proposition}
$\sigma \leq_{ib} \tau$ if and only if $\mathsf{obs}(\sigma) \leq_{os}
  \mathsf{obs}(\tau)$
\label{inequational}
\end{proposition}
\begin{proof}
Suppose $\sigma \leq_{ib} \tau$ and $S \in \mathsf{obs}(\sigma)$. Then
by \ref{topp}, $\sigma;\alpha_S = \top$. Then by assumption
$\tau;\alpha_S = \top$. Then by \ref{topp}, $S \supseteq T \in
\mathsf{obs}(\tau)$.

Conversely, suppose $\sigma;\alpha_S = \top$ for some
\textsf{O}-deterministic set $S$ (invoking \ref{obseqv}). So $S
\supseteq T \in \mathsf{obs}(\sigma)$ by \ref{topp}. Then since
$\mathsf{obs}(\sigma) \leq_{os} \mathsf{obs}(\tau)$, $T \supseteq R
\in \mathsf{obs}(\tau)$. So $S \supseteq R \in \mathsf{obs}(\tau)$. So
by \ref{topp} $\tau;\alpha_S = \top$, as required. \qed
\end{proof}

We have now shown that $\sigma =_{ib} \tau$ iff $\mathsf{ovw}(\sigma)
=_{os} \mathsf{ovw}(\tau)$. We shall now show that this is equality of
\textsf{O}-sets by showing that $\leq_{os}$ is antisymmetric for the
kind of sets we are dealing with.

We note that $\leq_{os}$ is \emph{not} antisymmetric on general sets
of \textsf{O}-deterministic sets. Let $\sigma = \{ \{ q_2 q_1 5_1, q_2
0_2 \}, \{ q_2 0_2 \} \}$ and $\tau = \{ \{ q_2 0_2 \} \}$ on the
arena $\mathbb{N} \rightarrow \mathbb{N}$. Each set in $\sigma$ and
$\tau$ are \textsf{O}-deterministic, and we have $\sigma =_{os} \tau$
with $\sigma \neq \tau$. However the strategy $\sigma$ does not come
from any innocent strategy, since the strategy would have to both
query and not query its argument. Thus we need to put further
restrictions on these sets of \textsf{O}-deterministic sets regarding
how the \textsf{O}-deterministic sets can interact with each other ---
a condition of determinacy.

\begin{definition}
An \emph{observational strategy} on $A$ consists of a set $\sigma$ of
\textsf{O}-deterministic sets over $A$ such that if $S, T \in \sigma$
with $S \neq T$ then there exists a play $t$ and \textsf{O}-moves
$o_1, o_2$ with $o_1 \neq o_2$ such that $to_1 \in S$ and $to_2 \in
T$.
\end{definition}

\noindent This says that if two \textsf{O}-deterministic sets differ,
then they first differ at an \textsf{O}-move.

\begin{proposition}
For each innocent strategy $\sigma$, $\mathsf{obs}(\sigma)$ is an
observational strategy.
\label{obstrat}
\end{proposition}
\begin{proof}
Suppose $\mathsf{ovw}(s) \neq \mathsf{ovw}(t)$. Then it follows that
$s \neq t$. Since $s, t \in \sigma$ they must first differ at an
\textsf{O}-move by the determinacy condition on strategies. Thus $ro_1
\sqsubseteq s$, $ro_2 \sqsubseteq t$ for $o_1 \neq o_2$. Then
$\llcorner ro_1 \lrcorner = \llcorner r \lrcorner o_1 \in
\mathsf{ovw}(s)$ and $\llcorner ro_2 \lrcorner = \llcorner r \lrcorner
o_2 \in \mathsf{ovw}(t)$ with $o_1 \neq o_2$ as required. \qed
\end{proof}

\begin{proposition} 
If $\sigma$ is an observational strategy, $S, T \in \sigma$ with $S
\subseteq T$ then $S = T$.
\label{neq}
\end{proposition}
\begin{proof}
Suppose $S \subseteq T$ and for contradiction that $S \neq T$. Then
there exists $t, o_1, o_2$ with $to_1 \in S$, $to_2 \in T$ and $o_1
\neq o_2$. But then $to_1 \in T$ since $S \subseteq T$. Thus $to_1,
to_2 \in T$ with $o_1 \neq o_2$. This contradicts
\textsf{O}-determinacy of $T$. \qed
\end{proof}

\noindent From this it is simple to show that $\leq_{os}$ is
antisymmetric:

\begin{proposition}
Let $\sigma$ and $\tau$ be observational strategies such that $\sigma
\leq_{os} \tau$ and $\tau \leq_{os} \sigma$. Then $\tau = \sigma$.
\label{antisym}
\end{proposition}
\begin{proof}
It will of course suffice to show wlog that $\sigma \subseteq
\tau$. Let $S \in \sigma$. Then since $\sigma \leq_{os} \tau$ we have
$T \subseteq S$ with $T \in \tau$. Then since $\tau \leq_{os} \sigma$
we have $S' \subseteq T$ with $S' \in \sigma$. Then $S' \subseteq S$
with both in $\sigma$ so it follows by \ref{neq} that $S = S'$. Since
$S \subseteq T \subseteq S$ it follows that $S = T$, i.e. $S \in \tau$
as required. \qed
\end{proof}

\noindent We can now show our main result.

\begin{theorem}
Two innocent strategies $\sigma$ and $\tau$ are observationally
equivalent if and only if $\mathsf{obs}(\sigma) = \mathsf{obs}(\tau)$.
\label{eqresult}
\end{theorem}
\begin{proof}
Suppose $\sigma =_{ib} \tau$. Then $\sigma \leq_{ib} \tau$ and $\tau
\leq_{ib} \sigma$. Then by \ref{inequational}, $\mathsf{obs}(\sigma)
\leq_{os} \mathsf{obs}(\tau)$ and $\mathsf{obs}(\tau) \leq_{os}
\mathsf{obs}(\sigma)$. But by \ref{obstrat} both
$\mathsf{obs}(\sigma)$ and $\mathsf{obs}(\tau)$ are observational
strategies. Thus by \ref{antisym} we have $\mathsf{obs}(\sigma) =
\mathsf{obs}(\tau)$.

Conversely if $\mathsf{obs}(\sigma) = \mathsf{obs}(\tau)$ then
$\mathsf{obs}(\sigma) \leq_{os} \mathsf{obs}(\tau)$ and
$\mathsf{obs}(\sigma) \leq_{os} \mathsf{obs}(\tau)$ since $\leq_{os}$
is clearly reflexive. Then $\sigma \leq_{ib} \tau$ and $\tau \leq_{ib}
\sigma$ by \ref{inequational} so $\sigma =_{ib} \tau$ as
required. \qed
\end{proof}

To return to our example, it is easy to see that
$\mathsf{obs}(\mathsf{add_{LR}}) = \mathsf{obs}(\mathsf{add_{RL}})$ ---
and the same result is obtained if we consider $\mathsf{obs}$ of any
other $\mathsf{add}$ strategy (e.g. interrogation of arguments
multiple times). We precisely forget repetition and ordering in this
construction, and thus only represent ``purely functional'' behaviour.

\section{A Fully Abstract Model?}


\noindent We may use the above result to formulate a fully abstract
model for PCF.

\begin{definition}
  We define the category $\mathcal{OBS_L}$. Objects of
  $\mathcal{OBS_L}$ are games. An arrow $s : A \rightarrow B$ is a set
  of sets of \textsf{O}-views of plays over the game $A \Rightarrow B$
  such that $\sigma = \mathsf{obs}(\tau)$ for some innocent strategy
  $\sigma_s : A \Rightarrow B$. The identity $e$ for an object $A$ is
  given by $\mathsf{obs}(\mathsf{id_A})$ where $\mathsf{id}_A$ is the
  copycat strategy on the game $A$. If $s : A \rightarrow B$ and $t :
  B \rightarrow C$, we define composition $s;t$ as the observational
  strategy given by $\mathsf{obs}(\sigma_s; \sigma_t)$.
\end{definition}

\noindent We can show that composition in $\mathcal{OBS_L}$ is
well-defined via the following proposition, following from
\ref{eqresult} and results in \cite{AMc_GS}.

\begin{proposition}
  If $\sigma_1, \sigma_2 : A \rightarrow B, \tau : B \rightarrow C$
  are innocent strategies with $\sigma_1 =_{ib} \sigma_2$ then
  $\sigma_1;\tau =_{ib} \sigma_2;\tau$. Similarly if $\sigma : A
  \rightarrow B, \tau_1, \tau_2 : B \rightarrow C$ with $\tau_1 =_{ib}
  \tau_2$ then $\sigma;\tau_1 =_{ib} \sigma;\tau_2$.
\end{proposition}

We can see that $\mathcal{OBS_L}$ is indeed a category by appealing to
associativity and identity in the category $\mathcal{C}_{inn}$. We can
then give a denotation of PCF in this category --- the denotation of
types are the same as that for the game semantic model, and the
denotation of a term $S$ is given by $\mathsf{obs}(\llbracket S
\rrbracket_{ib})$.

The above treatment gives a concrete fully abstract ``model'' of PCF,
but it doesn't give us any extra information about how the terms of
PCF look denotationally. In particular it would be good to define
precisely which observational strategies come from an innocent
strategy, and to define their composition directly --- this would
explicitly yield a categorical model which is full abstract for
PCF. This seems possible, but many details need checking; this is left
for future work.

\bibliographystyle{alpha}
\bibliography{../puregames/mybib.bib}

\begin{thebibliography}{AJM95}

\bibitem[AJM95]{AJM_PCF}
Samson Abramsky, Radha Jagadeesan, and Pasquale Malacaria.
\newblock Full abstraction for pcf.
\newblock {\em Information and Computation}, 163:409--470, 1995.

\bibitem[AM99]{AMc_GS}
S.~Abramsky and G.~McCusker.
\newblock Game semantics.
\newblock In H.~Schwichtenberg and U.~Berger, editors, {\em Computational
  Logic: Proceedings of the 1997 Marktoberdorf Summer School}, pages 1--56.
  Springer-Verlag, 1999.

\bibitem[HO00]{HO_PCF}
J.~M.~E. Hyland and C.-H.~L. Ong.
\newblock On full abstraction for pcf: I, ii, and iii.
\newblock {\em Inf. Comput.}, 163(2):285--408, 2000.

\bibitem[Loa01]{Loa_PCF}
Ralph Loader.
\newblock Finitary pcf is not decidable.
\newblock {\em Theor. Comput. Sci.}, 266(1-2):341--364, 2001.

\bibitem[McC96]{McC_FPC}
G.~McCusker.
\newblock Games and full abstraction for fpc.
\newblock In {\em Logic in Computer Science, 1996. LICS '96. Proceedings.,
  Eleventh Annual IEEE Symposium on}, pages 174--183, Jul 1996.

\bibitem[Plo77]{Plot_LCF}
G.~D. Plotkin.
\newblock Lcf considered as a programming language.
\newblock {\em Theoretical Computer Science}, 5(3):223 -- 255, 1977.

\end{thebibliography}

\end{document}